\documentclass[lnicst]{svmultln}
\usepackage{makeidx}  
\usepackage{graphicx}
\usepackage{amssymb, amsmath}
\usepackage{epstopdf}
\usepackage{fancybox}
\usepackage{xcolor}
\newcommand{\x}{\mathbf{x}}
\newcommand{\y}{\mathbf{y}}
\newcommand{\cvec}{\mathbf{c}}
\newcommand{\bvec}{\mathbf{b}}

\newcommand{\ie}{{\it i.e.}, }   
\newcommand{\eg}{{\it e.g.}, }  
\newcommand{\cf}{{\it cf.}, }    
\newcommand{\etal}{\emph{et al. }}
\newcommand{\np}{\mathrm{NP}}

\newcommand{\R}{\mathbb{R}}  
\newcommand{\N}{\mathbb{N}}  

\begin{document}
\mainmatter              
\title{Strategic Seeding of Rival Opinions\thanks{This research was conducted while S. D. Johnson was a graduate student at the University of California, Davis.}}
\titlerunning{On the Strategic Seeding of Opinion Dynamics}  
%
\author{Samuel D. Johnson\inst{1} \and Jemin George\inst{2} \and Raissa M. D'Souza\inst{3}}
\authorrunning{S. D. Johnson, J. George, and R. M. D'Souza}   
%
\tocauthor{Samuel D. Johnson, Jemin George, Raissa M. D'Souza}
\institute{HRL Laboratories, LLC., Malibu, CA 90265, USA\\
\email{sdjohnson@hrl.com}
\and
United States Army Research Laboratory,
Adelphi, MD 20783, USA\\
\email{jemin.george.civ@mail.mil}
\and
University of California, Davis, CA 95616, USA\\
\email{raissa@cse.ucdavis.edu}}

\maketitle              

\begin{abstract}        
We present a network influence game that models players strategically seeding the opinions of nodes embedded in a social network. A social learning dynamic, whereby nodes repeatedly update their opinions to resemble those of their neighbors, spreads the seeded opinions through the network. After a fixed period of time, the dynamic halts and each player's utility is determined by the relative strength of the opinions held by each node in the network \emph{vis-\`{a}-vis} the other players. We show that the existence of a pure Nash equilibrium cannot be guaranteed in general. However, if the dynamics are allowed to progress for a sufficient amount of time so that a consensus among all of the nodes is obtained, then the existence of a pure Nash equilibrium can be guaranteed. The computational complexity of finding a pure strategy best response is shown to be $\np$-complete, but can be efficiently approximated to within a $(1 - 1/e)$ factor of optimal by a simple greedy algorithm.

\end{abstract}

\section{Introduction}
\label{sec:intro}
Opinions are shaped by the information individuals obtain through their social connections. These opinions inform our career decisions, political views, and purchasing behaviors, which can ultimately spread to affect an entire society. The \emph{influence maximization problem} \cite{Kempe2003b} asks: Given the ability to seed a small number of individuals (nodes) in a social network to adopt a desired behavior (\eg to purchase a product, support a political candidate, contract an infection, etc.), which nodes should be selected so that the behavior subsequently spreads to a maximum fraction of the entire population? 

Strategic extensions to the influence maximization problem that involve two or more players representing competing, substitutable behaviors (products, opinions, infections, etc.) where each player selects a set of seed nodes so as to maximize the fraction of the population that eventually adopts their represented behavior, have subsequently been studied; see, for example,  \cite{Bharathi2007,Alon2010b,Borodin2010,Tzoumas2012,Fotakis2014,Goyal2014}. A common feature among most of these models is that the diffusion dynamics they use involve nodes making binary decisions to determine whether to fully adopt one of the diffused opinions. Furthermore, these models often employ progressive SI-style dynamics in which nodes that do not have an opinion must first make an irreversible decision to fully commit themselves to a single opinion in order to become a participant in the subsequent propagation of their adopted opinion.
Although such dynamics meet the desiderata of many applications with practical importance, we believe them to be unnatural for the diffusion of opinions through social networks because of the requirement that opinions are only allowed to pass through nodes that have already committed irreversibly to a particular preference.

In the current work, we present the \emph{network influence game} (NIG) to model the strategic seeding of opinions in a social network by $m \geq 2$ players and employs a dynamic that is arguably more appropriate than the SI-style dynamics used elsewhere. In our model, each node maintains a vector of opinions toward the $m$ players, and in each discrete time step, a node's opinion is updated to reflect a weighted averaging of their previous opinion and the opinions of their neighbors. These dynamics proceed in accordance with the \emph{consensus dynamic} \cite{DeGroot1974}, and terminate after a period of $T$ time steps. 
The consensus dynamic has played a prominent role in the study of opinion dynamics in social networks -- see, for example,  the surveys by Jackson \cite{Jackson2008,Jackson2011c}. 
In our competitive model, players' utilities are defined to be a function of each node's relative opinions toward the players upon the conclusion of the dynamic process. 

The NIG model is formally presented in Section~\ref{sec:model}, followed by a discussion on the nature of influence pertaining to it in Section~\ref{sec:model:influence}. 
Section~\ref{sec:nash} contains our results on the existence of pure Nash equilibrium strategies. We show that existence can be guaranteed if the dynamics run until a consensus opinion is reached, but cannot be guaranteed otherwise. 
In Section~\ref{sec:br} we show that computing a pure strategy best response is $\np$-complete, yet can be efficiently approximated to within a factor of $(1 - 1/e)$ of optimal.
Section~\ref{sec:con} concludes the paper with a discussion and suggested topics for future research.

\section{Model}
\label{sec:model}
The \emph{network influence game} NIG is specified by a tuple $\langle M, G, \Delta_C, T, \bvec, \pi \rangle$ with a player set, $M = \{1, \dots, m\}$; a network, $G = (V,E)$, represented by a weighted digraph with $|V| = n$ nodes; the opinion dynamic, $\Delta_C$, that determines how influence spreads between adjacent nodes in the network; the length of time, $T \in \N$, that the diffusion dynamic is allowed to proceed; a profile, $\bvec = (b_1, \dots, b_m)$, of integer seed budgets $b_i > 0$ for each player $i \in M$; and a utility function, $\pi_i(\cdot)$, that aggregates the nodes' opinions upon the conclusion of the diffusion dynamics at time $T$ into a non-negative payoff for each player $i \in M$. 

We assume that the graph $G$ is strongly connected. Furthermore, we stipulate that for all edges $(u,v) \in E$, the edge weight $w(u,v) > 0$ and, for all nodes $v \in V$, the sum of all incoming edge weights equals one.

A (pure) strategy for player $i$ is a subset $s_i \subseteq V$ of at most $b_i$ seed nodes. A strategy profile $s = (s_1, \dots, s_m)$ specifies the seed nodes chosen by all $m$ players, and is used to set the initial conditions for the opinion dynamics, the result of which determines the players' utilities. 

The NIG's opinion dynamics, $\Delta_C$, is based on DeGroot's \emph{consensus dynamic} \cite{DeGroot1974}. For this dynamic, each node $v \in V$ maintains a length-$m$ opinion vector $\x^v = (x^v_1, \dots, x^v_m)$ with the entry $0 \leq x^v_i \leq 1$ representing $v$'s opinion toward player $i$. If $x^v_i > x^v_j$, then this is to be interpreted as $v$ holding a more favorable opinion toward player $i$ than toward player $j$. We require that the sum of a node's opinions is at most one; $|\x^v|_1 = \sum_{i \in M} x^v_i \leq 1$. Since we need to refer to the evolution of a node's opinion over time, we will use $\x^v(t)$ to denote $v$'s vector of opinions at time $t$, with entry $x^v_i(t)$ representing $v$'s opinion toward player $i$ at time $t$.

The diffusion process is initialized by $s$ at time $t=0$. This involves each player $i$ implanting a ``seed opinion'' $\y^i$ into each of the nodes included in $s_i$. The seed opinions $\y^i = (y^i_1, \dots, y^i_m)$ are defined as $y^i_i = 1$ and $y^i_j = 0$ for all $j \neq i$, meaning that $\y^i$ specifies a high (maximum) opinion toward player $i$ and a low (minimum) opinion toward all other players $j$. Let $M_v(s) = \{i \mid v \in s_i\} \subseteq M$ denote the subset players that include a given node $v$ in their strategy, and define $m_v(s) = | M_v(s) |$. Finally, let $V_s$ denote the subset of nodes that are designated as seed nodes by at least one player. We initialize the opinions of the nodes $v \in V$ towards each player $i \in M$
	\begin{equation}
	\label{eq:init}
	x^v_i(0) =	\begin{cases}
				\frac{1}{m_v(s)} \sum_{j \in M_v} y^j_i & \text{ if } v \in V_s \\
				\varepsilon & \text{ if } v \in V \setminus V_s,
			\end{cases}
	\end{equation}
where $0 < \varepsilon \ll 1/m$ is a small constant. Equation~\eqref{eq:init} specifies that each seed node $v \in V_s$ is initialized to the average seed opinion of the players that include $v$ in their strategies. Otherwise, for a node $v \notin V_s$, initialization involves assigning a $\varepsilon$ opinion value toward every player $i \in M$. 

The consensus dynamic, $\Delta_C$, proceeds in discrete time steps $t = 1, 2, \dots, T$ and specifies that the opinion of a node $v$ at time $t$ is a weighted average of its prior opinion and the opinions of its neighbors at time $t-1$. Specifically, 
	\begin{equation}
	\label{eq:dynamics}
	x^v_i(t) = (1 - \alpha) \cdot x^v_i(t-1) + \alpha \sum_{u \in N^+(v)} w(u,v) \cdot x^u_i(t-1), \qquad \forall i \in M,
	\end{equation}
where $N^+(v) = \{u \mid (u,v) \in E\}$ is the set of $v$'s incoming neighbors in $G$ and $0 < \alpha < 1$ is a model parameter. 

This dynamic can be expressed more succinctly in matrix form. Let $A$ be the weighted adjacency matrix of $G$ with entries $a_{ij} = w(i,j)$, and define $\Gamma = (1 - \alpha) I + \alpha A^\intercal$ to be the $n \times n$ influence matrix\footnote{$\Gamma$ is sometimes referred to as a \emph{listening structure} \cite{DeMarzo2003} or \emph{interaction matrix} \cite{Golub2010}.} with entries $\gamma_{vu}$ conveying the amount of direct influence that the opinions of node $u$ shape those of node $v$ from one time step to the next. By construction, $\Gamma$ is an aperiodic stochastic matrix.  Let $\x_i(t)$ denote the length-$n$ vector containing entries for each node's opinion toward player $i$ at time $t$. Using $\Gamma$ and $\x_i(\cdot)$, we can rewrite the dynamics in Equation~\eqref{eq:dynamics}  as $\x_i(t) = \Gamma \x_i(t-1)$. In particular, the opinions toward player $i$ after $T$ time steps is simply $\x_i(T) = \Gamma^T \x_i(0)$.\footnote{We use $\Gamma^T$ to denote the matrix $\Gamma$ raised to the $T$th power. For matrix transposition, we use the notation $\Gamma^\intercal$.} 
	
Upon the termination of the opinion dynamics after $T$ steps, each node $v$ is left holding an opinion vector $\x^v(T)$. The utility for player $i \in M$ is defined to be the average relative opinion held by the population toward $i$,
	\begin{equation}
	\label{eq:utility}
	\pi_i(s) = \frac{1}{n} \sum_{v \in V} \frac{x^v_i(T)}{ |\x^v(T)|_1 }.
	\end{equation}
Notice that, \eqref{eq:utility} implies that for any strategy profile $s$, we have $\sum_{i \in M} \pi_i(s) = 1$, so the NIG is a constant-sum game.

\subsection{Influence}
\label{sec:model:influence}
A player's best response strategy in the NIG involves selecting an ``influential'' subset of seed nodes so that, upon the termination of the dynamics, the average relative opinion held by the nodes toward the player is maximized. The precise character of \emph{influence} in this context deserves some examination.

It is well-known that the DeGroot consensus dynamic converges so that a common opinion is shared by every node in the network is guaranteed as $T \rightarrow \infty$ when the matrix $\Gamma$ is aperiodic and stochastic. The consensus obtained is described by a weighted sum of the nodes' initial opinions (at time $t=0$), with the weights given by the entries in the eigenvector of $\Gamma$ corresponding to the eigenvalue $1$. Hence, for sufficiently large $T$, a node's influence is directly related to their eigenvector centrality. However, if $T$ is not large enough to obtain consensus, then a node's eigenvector centrality no longer corresponds to the influence they exert on the average opinions upon the termination of the dynamics. 

The importance of $T$ in characterizing the influence that a node exerts on a diffusion process was recently identified in an empirical study by Banerjee \etal \cite{Banerjee2013} (see also \cite{Banerjee2014}), which led them to define a quantity called \emph{diffusion centrality}. 
Although their definition corresponds to a different diffusion dynamic than the one we consider in this paper, we can still offer a definition that is qualitatively similar to theirs but tailored to the consensus dynamic.

Let $\delta[j]$ denote the $n$-dimensional column vector consisting of a one in row $j$ and zeros everywhere else. The diffusion centrality of a node $v$ is defined to be the vector $\cvec^v = \Gamma^T \delta[v]$ whose entries $c^v_u$ describe the fraction of node $u$'s opinion at time $T$ that is due to node $v$'s initial opinion at time $t=0$. 
As $T \rightarrow \infty$, the convergence of the consensus dynamic ensures that $|c^v_{u} - c^v_{u'}| \rightarrow 0$ for all $u, u' \in V$; let $c^v$ denote this uniform amount of influence that $v$'s brings to bear upon the final opinion of every node in $V$. The value $c^v$ is precisely the $v$'th entry of the unique left eigenvector $\cvec$ of $\Gamma$ corresponding to the eigenvalue $1$, and the sum of the entries in $\cvec$ equal $1$.

\section{On the Existence of Pure Nash Equilibrium}
\label{sec:nash}
This section presents our results on the existence of pure Nash equilibrium strategies. Recall that a pure strategy profile $s = (s_i, s_{-i})$ is a Nash equilibrium if, for every player $i \in M$ and every possible pure strategy $s_i'$ for that player, we have $\pi_i(s_i, s_{-i}) \geq \pi_i(s_i', s_{-i})$. 
Throughout this paper, all of the results regarding Nash equilibrium will be with respect to pure strategies, and all unqualified mentions of \emph{strategy} should be understood to refer to a \emph{pure strategy}; all mentions of \emph{Nash equilibrium} refer to \emph{pure Nash equilibrium}.

\subsection{At Consensus}
\label{sec:nash:con}
In this section we establish the existence of pure strategy Nash equilibria for NIGs in which $T$ is sufficiently large to ensure that the opinions reach a consensus. In the consensus regime, all nodes share the same final opinion vector, $\x(T) = (x_1(T), \dots, x_m(T))$, and the utility for each player $i \in M$ is simply
	\begin{equation}
	\label{eq:utility:steadystate}
	\pi_i(s) =  \frac{x_i(T)}{ |\x(T)|_1}.
	\end{equation}

With the profile of weights $\cvec = (c^{v_1}, c^{v_2}, \dots, c^{v_n})$, where $c^v$ denotes the weight of node $v$'s contribution to this consensus, we can express the consensus opinion toward player $i$ as $x_i(T) = \sum_{v \in V} c^v x^v_i(0)$. 
Notice that the consensus opinion $x_i(T)$ for player $i$ is a countably additive function of player $i$'s strategy, $s_i$. This implies that in order for a deviation from $s_i$ to $s_i'$ to increase the consensus opinion toward player $i$, then any deviation from $s_i$ to $s_i''$ where $s_i'' = (s_i \setminus \{v\}) \cup \{u\}$, $v \in s_i \setminus s_i'$, and $u \in s_i' \setminus s_i$ will also increase the consensus opinion toward player $i$. In other words, if $i$ can increase their share of the consensus opinion by swapping $k$ nodes in their strategy, then they can also gain from swapping only a single node. 

A key feature to the consensus case is that, since all nodes will converge to the same opinion, there is no longer any need for players to make strategic trade-offs involving increasing their opinion share among one particular subset of nodes at the cost of decreasing their opinion share elswhere. This is reflected by the fact that the utility function in Equation~\eqref{eq:utility} reduces to Equation~\eqref{eq:utility:steadystate} in the consensus regime (\ie when $T$ is sufficiently large).

\begin{proposition}
\label{prop:nash:steadystate}
Every NIG in which $T$ is sufficiently large so as to ensure that a consensus opinion $\x^v(T)$ is reached among every node $v \in V$ has at least one pure Nash equilibrium.
\end{proposition}

\begin{proof}[Sketch]
Assume that the $m$ players are ordered so that $b_1 \geq b_2 \geq \cdots \geq b_m$ and the nodes $V = \{ v_1, v_2, \dots, v_n \}$ are ordered so that $c^{v_1} \geq c^{v_2} \geq \cdots \geq c^{v_n}$. Set player $1$'s seed strategy $s_1$ such that it maximizes $x_1(T)$. The strategies for players $i = 2, 3, \dots, m$ will be built sequentially, so that $s_i$ is a best response to the profile $s^{i-1} = (s_1, \dots, s_{i-1})$. We then argue that, given $s^i = (s_1, \dots, s_{i-1}, s_i)$, for every player $j < i$, $s_j$ is a best response to $s^i_{-j} = (s_1, \dots, s_{j-1}, s_{j+1}, \dots, s_i)$. Here, we will give the proof for $m=2$ players; the proof for $m \geq 2$ players follows from an inductive argument that is based on similar reasoning.

Set $s_1 = \{v_1, \dots, v_{b_1} \}$, and let $s_2$ be a best response to $s_1$ that maximizes the value of the consensus opinion toward player $2$, $x_2(T)$. We have two cases to consider: i) $s_1 \cap s_2 = \emptyset$, and ii) $s_1 \cap s_2 \neq \emptyset$.

The first case is trivial: if $s_2$ does not contain any of the nodes in $s_1$, then player $1$ enjoys exclusive access to $b_1$ of the most influential seed nodes.

For the second case, let $r = s_1 \cap s_2$ be the set of seed nodes chosen by player $2$ that are also in $s_1$. Suppose, toward a contradiction, that player $1$ can strictly benefit from changing to a strategy $s_1' = (s_1 \setminus \{u\}) \cup \{v\}$ for some nodes $u \in r$ and $v \in V \setminus (s_1 \cup s_2 )$; \ie player $1$ swaps out a shared node $u$ for exclusive access to another node $v$. By swapping out $u$ for $v$, player $1$ may increase the consensus opinion toward them self, but they would also be increasing the consensus opinion toward player $2$ by relinquishing their share of the influence weight $c^u$. But, by virtue of player $2$'s inclusion of $u$ instead of $v$ in their own best response, it must be the case that the relative opinion toward player $1$ -- and, thus, player $1$'s utility -- would not improve by swapping $u$ for $v$. Therefore, $\pi_1(s_1', s_2) \leq \pi_1(s_1, s_2)$, contradicting the claim that $s_1'$ can earn player $1$ a strictly higher utility than $s_1$. \qed
\end{proof}

\subsection{The General Case}
\label{sec:nash:noexist}
In this section, we show that the existence of a pure Nash equilibrium is \emph{not} guaranteed in NIGs when the dynamic does not reach a consensus opinion. We prove the non-existence for the symmetric setting, in which every player shares the same seed budget; a proof for the asymmetric setting can be found in the full version of this paper.

\begin{proposition}
\label{prop:ne:cycle}
For any $m \geq 2$ and symmetric seed budget $b$, there exist NIGs that do not admit a pure strategy Nash equilibrium.
\end{proposition}

\begin{figure}[t]
	\centering
	\includegraphics[scale=0.7]{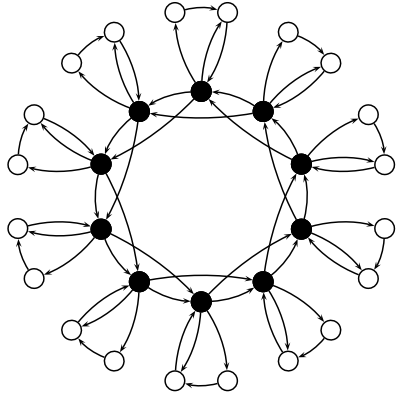}
	\caption{Example of the graph construction in the proof of Proposition~\ref{prop:ne:cycle} with $m = 3$ players and budgets $b = 2$. The dark nodes are the ``central nodes'' and the white ones are the ``petal nodes''. This network does not admit a pure Nash equilibrium.}
	\label{fig:no-nash-graph}
\end{figure}

\begin{proof} 
A construction that does not admit a pure Nash equilibrium is as follows: Create $\mu = m(b+1) + 1$ nodes $v_0, v_1, \dots, v_{\mu-1}$, and add directed edges from each $v_i$ to $v_{i+k}$ for $k = 1, \dots, b$. Next, for each $v_i$, add two additional ``petal nodes'' $v_{i,l}$ and $v_{i,r}$, and add the following four links: $(v_i, v_{i,l})$, $(v_i, v_{i,r})$, $(v_{i,l}, v_{i,r})$, and $(v_{i,r}, v_i)$. Define the weight of each edge $(u,v)$ to be $w(u,v) = \frac{1}{|N^+(v)|}$. Note that the sum of each node's incoming edge weights equals one, and that the resulting graph is strongly connected. (See Figure~\ref{fig:no-nash-graph} for an example.)

Let $G$ be a graph that implements the above construction, and set the parameters $\alpha = 1/2$ and $T = b$. Because a player's optimal strategy will never involve seeding a petal node, we restrict our attention to strategies comprising only central nodes. Let $V_C$ denote these $\mu$ central nodes. Since $|V_C| = \mu = m(b + 1) + 1$ and the total seed budget among all $m$ players is $bm$, it will always be the case that a player's best response will not include nodes that are already included in another player's strategy, and exactly $m + 1$ of the central nodes will not be included as seed nodes in any players' strategies. Let $U \subset V_C$ denote these $m+1$ unseeded central nodes. By design, the graph is constructed so that players prefer selecting seed nodes with successors in $U$. Of the $|U| = m+1$ nodes, at most $m$ of them will have no predecessors that are also in $U$; and there will always be at least one node $u^* \in U$ that does have a predecessor in $U$. Let $u'$ denote $u^*$'s predecessor in $U$. A player $i$ with a seed node $v \in s_i$ that is a successor of $u'$ can always improve their utility by changing to a strategy $s_i' = \{u'\} \cup (s_i \setminus \{v\})$. Since the existence of (at least one) such $u^*$ is guaranteed, then there will always be a player $i$ that can change strategies for an increase in utility. Hence, a Nash equilibrium cannot be obtained. \qed
\end{proof} 

\section{Computational Properties of Best Response}
\label{sec:br}

This section considers the computational problem of finding, for a given player $i$ and strategy profile $s_{-i}$, a strategy $s_i$ that maximizes $i$'s utility, $\pi_i(s_i, s_{-i})$. Such a strategy $s_i$ is called a \emph{best response} to the profile $s_{-i}$ of strategies belonging to every other player $j \neq i$. 

Our first result in this section establishes the computational complexity of the best response problem.

\begin{proposition}
\label{prop:complexity}
Finding a best response strategy for the NIG is $\np$-complete.
\end{proposition}

\begin{proof}[Sketch]
Hardness follows by reduction from \textsc{Set Cover} and completeness is due to the fact that the utility function can be computed in polynomial time.  \qed
\end{proof}

Next, we turn the problem of finding an approximate best response. It will be useful to adopt the following definition of the utility function, which is equivalent to Equation~\eqref{eq:utility}:
	\begin{equation}
	\label{eq:approx:1}
	\pi_i(s) = \frac{1}{n} \sum_{v \in V} g^i_v(s)
	\end{equation}
where
	\begin{equation}
	\label{eq:approx:2}
	g^i_v(s) = \frac{f^i_v(s)}{f_v(s)},
	\end{equation}
	\begin{equation}
	\label{eq:approx:3}
	f^i_v(s) = \sum_{u \in s_i} \frac{c^u_v}{m_u(s)}  + \sum_{u \in V \setminus s} \varepsilon c^u_v,
	\end{equation}
and
	\begin{equation*}
	f_v(s) = \sum_{j \in M} f^j_v(s).
	\end{equation*}
Here, we employ the notation $m_u(s)$ to denote the cardinality of the set $M_u(s) = \{ j \mid u \in s_j \} \subseteq M$. The quantities $c^u_v$, which measure the amount of influence that node $u$ exerts on node $v$ after $T$ time steps, was defined in Section~\ref{sec:model:influence}.

Our main result for this section establishes that the utility function $\pi_i(\cdot)$ is submodular,\footnote{A set function $f : \Omega \rightarrow \R$ is \emph{submodular} if, for every $X \subseteq Y \subset \Omega$ and element $x \in \Omega \setminus Y$, we have $f(X \cup \{x\}) - f(X) \geq f(Y \cup \{x\}) - f(Y)$.}
from which it follows from the classic result by Nemhauser, Wolsey, and Fisher \cite{Nemhauser1978} that a $(1 - 1/e) \approx 0.6321$ approximation can be computed using a greedy algorithm.

\begin{proposition}
\label{prop:submodular}
The utility function $\pi_i(\cdot)$ is monotonic and submodular.
\end{proposition}

\begin{proof}[Sketch]
The submodularity and monotonicity follow immediately from establishing that \eqref{eq:approx:3} is increasing and \eqref{eq:approx:2} is submodular and the fact that since $\pi_i(\cdot)$ is a function that is defined by a linear combination of submodular functions (\cf Equation~\eqref{eq:approx:1}), then $\pi_i(\cdot)$ is itself submodular. \qed
\end{proof} 

\section{Discussion and Future Work}
\label{sec:con}
This paper presented a model for the strategic seeding of opinion dynamics using the simple, well-studied DeGroot consensus dynamic. We established that the existence of pure Nash equilibria cannot be guaranteed if the dynamic is not allowed to run to consensus. The amount of time required for the dynamic to reach to consensus is known to be slower in networks with modular (homophilic) connectivity patterns \cite{Golub2010b}. This implies that in practice, strategic opinion seeding on real-world social networks, which often exhibit modular structures, should not assume that there will be enough time for the population to coalesce around a shared, consensus opinion; and, crucially, the individuals that appear to be attractive seeds in the steady state regime when a consensus is reached (those with high eigenvector centrality) may not be the best choice if the dynamics halt in the transient regime.

Our findings in Section~\ref{sec:br} on the computational problem of finding best response strategies are in alignment with similar competitive influence models that employ opinion dynamics that are more complex and less amenable to analytical tractability than the simplistic DeGroot consensus dynamic we use. For example, in a competitive extension of the probabilistic \emph{independent cascade} model of Kempe \etal \cite{Kempe2003b},  Bharathi \etal \cite{Bharathi2007} establish a $(1 - 1/e)$ approximation guarantee using monotone and submodularity arguments that extend those used in \cite{Kempe2003b} for the ``single-player'' setting (see also Mossel and Roch \cite{Mossel2010b}). However, some models have approximation guarantees that can be significantly worse than $(1 - 1/e)$. For example, Borodin \etal \cite{Borodin2010} show that competitive extensions of the \emph{threshold} diffusion model do not share the $(1 - 1/e)$ approximation guarantee established in \cite{Kempe2003b} for the ``single player'' optimization setting. They do, however, offer a variant of the threshold model that does admit a $(1 - 1/e)$ approximation guarantee. Similar $(1 - 1/e)$ approximation guarantees are established more recently in competitive influence maximization models by Goyal \etal \cite{Goyal2014} and Fotakis \etal \cite{Fotakis2014}.

Our analysis in Section~\ref{sec:nash} highlights the importance of the length of the diffusion process, $T$, in guaranteeing the existence pure strategy Nash equilibria. The identification of conditions that are sufficient to guarantee the existence of equilibria for small, non-consensus reaching values of $T$ is an interesting open problem. Related to this is an intriguing extension to the model that would allow players to not only choose \emph{which} nodes to seed, but also \emph{when} to seed them. In our preliminary investigations into this extension, we have observed in simulations that some graphs contain nodes whose influence ``peaks'' at a greater magnitude in the dynamic's transient regime than in the steady state. We believe that such an extension would also more closely model many real-world applications, such as political contests and advertising campaigns, where timing can be an important consideration.

\subsection*{Acknowledgements}
\label{sec:con:ack}
The authors gratefully acknowledge support from the US Army Research Office MURI Award No. W911NF-13-1-0340 and Cooperative Agreement No. W911NF-09-2-0053.

\bibliographystyle{plain}
\bibliography{bib}

\end{document}